\documentclass[pdflatex,sn-mathphys-num]{sn-jnl}

\usepackage{graphicx}%
\usepackage{multirow}%
\usepackage{amsmath,amssymb,amsfonts}%
\usepackage{amsthm}%
\usepackage{mathrsfs}%
\usepackage[title]{appendix}%
\usepackage{xcolor}%
\usepackage{textcomp}%
\usepackage{manyfoot}%
\usepackage{booktabs}%
\usepackage{algorithm}%
\usepackage{algorithmicx}%
\usepackage{algpseudocode}%
\usepackage{listings}%
\usepackage{upgreek} 
\usepackage{enumerate}
\usepackage{rotating}
\usepackage{array}
\usepackage{graphicx}   
\usepackage{caption}    
\usepackage{rotating} 
\usepackage{diagbox}  

\theoremstyle{thmstyleone}%
\newtheorem{theorem}{Theorem}
\newtheorem{lemma}[theorem]{Lemma}

\theoremstyle{thmstyletwo}%
\newtheorem{example}{Example}%
\newtheorem{remark}{Remark}%

\theoremstyle{thmstylethree}%

\begin{document}
	
	\title[Article Title]{Five Constructions of Asymptotically Optimal Aperiodic Doppler Resilient Complementary Sequence Sets with New Parameters}
	
	
	\author[1]{\fnm{Xuanyu} \sur{Liu}}\email{18372867620@163.com}
	
	\author*[2]{\fnm{Pinhui} \sur{Ke}}\email{keph@fjnu.edu.cn}
	
	\author[3]{\fnm{Zuling} \sur{Chang}}\email{zuling\_chang@zzu.edu.cn}
	
	\affil[1]{School of Mathematics and Statistics, Fujian Normal University, Fuzhou, Fujian 350117, China}
	
	\affil[2]{Key Laboratory of Analytical Mathematics and Applications (Ministry of Education), Fujian Normal University, Fuzhou, Fujian 350117, China}
	
	\affil[3]{School of Mathematics and Statistics, Zhengzhou University, Zhengzhou, Henan, 450001, China}


	\abstract{Sequences exhibiting favorable ambiguity function characteristics play a critical role in radar detection systems and modern mobile communication applications. As a newly developed sequence family, Doppler resilient complementary sequence sets (DRCSSs) can effectively suppress ambiguity function sidelobes by coherently combining the ambiguity functions of their constituent subsequences. The objective of this paper is to present five classes of asymptotically optimal aperiodic DRCSSs with novel parameters based on trace functions over finite fields and column orthogonal complex matrices. Compared with existing asymptotically optimal aperiodic DRCSSs in the literature, the proposed aperiodic DRCSSs deliver superior or novel parameters. Notably, for three families of the constructed aperiodic DRCSSs, the column sequence peak-to-average power ratio (PAPR) is upper bounded by $p$ by selecting suitable column orthogonal complex matrices.}

	
	
	
	\maketitle
	
	\section{Introduction}
 Sequences with favorable correlation properties have been widely applied in engineering practices, especially in wireless communications and radar sensing, where typical applications include synchronization, channel estimation, interference cancellation, target detection, and positioning \cite{C1,C2}. A perfect sequence set requires that the auto-correlation function of each sequence behaves as an impulse function, while the cross-correlation between any two distinct sequences is identically zero. However, such an ideal sequence set cannot exist due to the constraint imposed by the Welch bound \cite{C3}.
 
 To attain ideal correlation properties, complementary sequences have been proposed and extensively investigated. A two-dimensional array can be viewed as a complementary sequence. It is referred to as a perfect complementary sequence (PCS) if the sum of the auto-correlations of its row sequences is zero for all non-zero time shifts. Golay \cite{C4} first introduced the concept of aperiodic complementary sequence pairs in 1961. Later, in 1972, Tseng and Liu \cite{C5} generalized this concept to sets consisting of more than two sequences. Bomer et al. \cite{C6} and Luke \cite{C7} further developed periodic and odd-periodic PCSs, respectively.
 A perfect complementary sequence set (PCSS) consists of multiple perfect complementary sequences, such that the sum of their cross-correlations is zero \cite{C8}. PCSSs have become crucial in numerous applications, including multi-carrier code-division multiple-access (MC-CDMA) systems \cite{C9}, inter-symbol interference channel estimation \cite{C10}, radar waveform design \cite{C11}, and peak-to-average power ratio (PAPR) reduction \cite{C12}. In a PCSS, the set size $K$ is constrained by the flock size $M$, which denotes the number of subsequences in each complementary sequence, and $K \leq M$. In particular, if $K=M$, a PCSS is called a complete complementary code (CCC) \cite{C8}.
 Nevertheless, the limited set size of PCSSs presents a critical drawback in CDMA systems, where the number of users is strictly restricted by the number of subcarrier channels. To address this limitation, zero-correlation zone complementary sequence sets (ZCZ-CSSs) \cite{C13,C14}, low-correlation zone complementary sequence sets (LCZ-CSSs) \cite{C15}, and quasi-complementary sequence sets (QCSSs) \cite{C16,C17} have been proposed in the literature.	
 
 In high-mobility scenarios, modern communication and radar systems often confront the challenging Doppler effect. Traditional sequences, including complementary sequences, usually exhibit inadequate performance in such environments. Different from conventional sequence design criteria, this paper introduces the ambiguity function (AF) to characterize the receiver performance when both time delay and Doppler shift exist simultaneously. Ding et al. \cite{C18} first derived a theoretical bound on the magnitude of the ambiguity function (AF) for a sequence set. Such sequence sets are termed Doppler resilient sequence sets (DRSSs) \cite{C19}. Subsequently, in 2022, Ye et al. \cite{C19} improved the lower bound by exploiting the property that the magnitude of the auto-AF of a unimodular sequence at zero delay is zero for any non-zero Doppler shift. In addition, several DRSSs meeting the proposed periodic AF lower bound have been developed. More optimal DRSSs with low AF magnitudes have been reported in \cite{C19,C20,C21,C22,C23,C24}. However, the design of optimal DRSSs remains a challenging task.
 
 In order to reduce the magnitude of the AF, based on the idea of complementary sequences and DRSSs, Shen et al. \cite{C25} proposed the concept of a Doppler resilient complementary sequence set (DRCSS) and derived lower bounds for periodic, aperiodic, and odd-periodic AF of DRCSSs. Moreover, Shen et al. \cite{C25} introduced some constructions of DRCSSs based on algebraic and combinatorial methods such as orthogonal matrices, circular Florentine rectangles and difference sets, which can generate the optimal DRCSSs. Later, Shen et al. \cite{C26} derived several classes of asymptotically optimal and near-optimal DRCSSs by using one-coincidence frequency-hopping sequence sets, almost difference sets and repeating the CCCs. By selecting an appropriate weight vector, Wang et al. \cite{D1} derived a tighter lower bound than \cite{C25} for aperiodic DRCSSs. Moreover, Wang et al. \cite{D1} introduced a class of aperiodic DRCSS using the quasi-Florentine rectangles and Butson-type Hadamard matrices. Subsequently, Wang et al. \cite{D2} proposed several sets of aperiodic DRCSSs based on generalized quasi-Florentine rectangles and Butson-type Hadamard matrices. The proposed aperiodic DRCSSs are shown to be asymptotically optimal with respect to the lower bound in \cite{D1}.
 
 In \cite{a1}, Li et al. constructed three families of asymptotically optimal aperiodic QCSSs using trace functions over finite fields. Subsequently, Xiao et al. \cite{a2} constructed three classes of asymptotically optimal aperiodic QCSSs using additive characters over finite fields. Recently, based on column orthogonal complex matrices, additive characters of finite fields and Butson-type Hadamard matrices, Wang et al. \cite{a3} constructed four families of asymptotically optimal aperiodic QCSSs. Motivated by the aforementioned works,
in this paper, based on trace functions over finite fields and column orthogonal complex matrices, we present five classes of aperiodic DRCSSs with novel parameters, all of which are asymptotically optimal with respect to the bound given in \eqref{eqD1}. Specifically, in Theorem \ref{theorem1}, we present a family of asymptotically optimal aperiodic DRCSSs with an alphabet size $p$ based on column orthogonal complex matrices. Compared with the existing aperiodic DRCSS families in \cite{C25,D1}, the proposed design achieves a larger set size under the same flock size, constituent sequence length, and maximum ambiguity magnitude. Meanwhile, in Theorem \ref{theorem2}, we construct another family of asymptotically optimal aperiodic DRCSSs with alphabet size $p$ via column orthogonal complex matrices, which also yields a larger set size than the design in \cite{D1,D2} under the same parameter settings. Furthermore, in Theorems \ref{theorem1}, \ref{theorem2}, and \ref{theorem3}, we prove that by choosing appropriate column orthogonal matrices, the column-sequence PAPR of each complementary matrix in the proposed DRCSSs is upper bounded by $p$. 
The parameters of the known families of aperiodic DRCSSs and those proposed in this paper are summarized in Table \ref{table0}. 

The rest of this paper is organized as follows. Section \ref{s2} introduces the necessary notations and lemmas. In Section \ref{s3}, we present five classes of DRCSSs and calculate their parameters. Section \ref{s4} concludes this paper.  

\begin{sidewaystable}[htbp]
	\centering
	\renewcommand{\arraystretch}{1.2} 
	\caption{The parameters of known aperiodic DRCSSs}
	\label{table0}
	\begin{tabular}{c|c|c|c|c|c|c|>{\raggedright\arraybackslash}m{4cm}|c}
		\hline  
		Set size   & Flock size  & Length  & $\hat{\theta}_{\text{max}}$ &Alphabet size& $Z_x$ & $Z_y$ & Constraints & References\\
		\hline
		$\tilde{F}(N)$ & $N$ & $N$ & $N$ & $N$ & $N$ & $N$ & $N\geq 2$ is an integer, $\tilde{F}(N)$ is the maximum number of rows for which an $\tilde{F}(N) \times N$ circular Florentine rectangles exist.   &  \cite{C25}\\
		\hline
		$F_{Q,1}(N)$ & $N$ & $N-1$ & $N$ & $r$ & $N-1$ & $N-1$ & $2\leq r \leq N$ is an integer, $N\geq 2$ is an integer, $F_{Q,1}(N)$ is the maximum number of rows for which an $F_{Q,1}(N) \times (N-1)$ quasi-Florentine rectangles exist. & \cite{D1} \\
		\hline
		$F(N)$ & $N$ & $N$ & $N$ & $r$ & $N$ & $N$ & $2\leq r \leq N$ is an integer, $N\geq 2$ is an integer, $F(N)$ is the maximum number of rows for which an $F(N) \times N$ Florentine rectangles exist. &  \cite{D1} \\
		\hline
		$F_{Q,N-L}(N)$ & $N$ & $L$ & $N$ & $r$ & $L$ & $L$ & $2\leq r \leq N$ is an integer, $N\geq 2$ is an integer, $F_{Q,N-L}(N)$ is the maximum number of rows for which an $F_{Q,N-L}(N) \times L$ quasi-Florentine rectangles exist, $2\leq L \leq N$. &  \cite{D2}\\
		\hline
		$q+1$ & $q$ & $q$ & $q$ & $p$ & $q$ & $q$ & $q=p^n>2$ &this paper\\
		\hline
	    $q+1$ & $q$ & $q-1$ & $q$ & $p$ & $q-1$ & $q-1$ & $q=p^n>3$ &this paper\\
		\hline
	    $q-1$ & $q$ & $q+1$ & $q$ & $p$ & $q+1$ & $q+1$ & $q=p^n>3$ &this paper\\
		\hline
		$q-1$ & $q-1$ & $q$ & $q-1$ & $q-1$ & $q$ & $q$ & $q=p^n>3$ &this paper\\
		\hline
        $q-1$ & $q-1$ & $q-1$ & $q-1$ & $q-1$ & $q-1$ & $q-1$ & $q=p^n>4$ &this paper\\
	\end{tabular}
\end{sidewaystable} 

	\section{Preliminaries}\label{s2}
	In this section, we present the basic notations and key concepts which will be employed throughout this work.
	
	\begin{itemize}
		\item $\mathbb{Z}_{N}$ denotes the residue class ring modulo $N$ and $\mathbb{Z}_N^*$ denotes the multiplicative group of units in $\mathbb{Z}_{N}$.  
		\item Let $q=p^n$, where $p$ is a prime and $n$ is a positive integer, $\mathbb{F}_q$ denotes the finite field with $q$ elements and $\mathbb{F}_q^*=\mathbb{F}_q\setminus\{0\}$.
		\item $\xi_N=e^{\frac{2\pi\sqrt{-1}}{N}}$ is a primitive $N$-th complex root of unity.
		\item $(\cdot)^*$ denotes the complex conjugate.
		\item $(\cdot)^{\mathrm{T}}$ denotes the transpose of a matrix.
	\end{itemize}

\subsection{Ambiguity Function}
Let $\mathbf{a}=(a(0),a(1),\cdots,a(N-1))$ and $\mathbf{b}=(b(0),b(1),\cdots,b(N-1))$ be two complex unimodular sequences with period $N$, i.e., $|a(i)|=1$ and $0\leq i<N$. The aperiodic cross-ambiguity function of $\mathbf{a}$ and $\mathbf{b}$ in time shift $\tau$ and Doppler shift $v$ is defined as follow:

$$\widehat{AF}_{\mathbf{a},\mathbf{b}}(\tau,v)= \begin{cases}
	\sum_{i=0}^{N-1-\tau}a(i) b(i+\tau)^*\xi_N^{i v}, & 0\leq \tau \leq N-1,\\
		\sum_{i=-\tau}^{N-1}a(i) b(i+\tau)^*\xi_N^{i v}, & 1-N \leq \tau < 0,\\
		0, & |\tau| \geq 0.
\end{cases}$$
If $\mathbf{a}=\mathbf{b}$, $\widehat{AF}_{\mathbf{a},\mathbf{b}}(\tau,v)$ is called the aperiodic auto-ambiguity function, which is denoted by $\widehat{AF}_{\mathbf{a}}(\tau,v)$. Especially, when $v=0$, the two-dimensional AF degenerates into a one-dimensional correlation function, with aperiodic auto- and cross-correlation function denoted as $\widehat{CF}_{\mathbf{a}}$ and $\widehat{CF}_{\mathbf{a},\mathbf{b}}$, respectively.

\subsection{Aperiodic Doppler Resilient Complementary Sequence Set}
A DRCSS $\mathcal{C}=\{\mathbf{C}^{0},\mathbf{C}^{1},\cdots,\mathbf{C}^{K-1}\}$ contains $K$ (i.e., set size) DRCSs, each of which consists of $M$ (i.e., flock size) subsequences of length $N$, i.e.,
$$
\mathbf{C}^{k}=
\begin{bmatrix}
	\mathbf{c}_0^{k} \\
	\mathbf{c}_1^{k} \\
	\vdots \\
	\mathbf{c}_{M-1}^{k}
\end{bmatrix}
=
\begin{bmatrix}
	c_{0,0}^{k} & c_{0,1}^{k} & \cdots & c_{0,N-1}^{k} \\
	c_{1,0}^{k} & c_{1,1}^{k} & \cdots & c_{1,N-1}^{k} \\
	\vdots & \vdots & \ddots & \vdots \\
	c_{M-1,0}^{k} & c_{M-1,1}^{k} & \cdots & c_{M-1,N-1}^{k}
\end{bmatrix},
$$
where $\mathbf{c}_m^{k}=(c_{m,0}^{k},c_{m,1}^{k},\cdots,c_{m,N-1}^{k})$, $0\leq k<K$, and $0 \leq m<M$.

For two DRCSs $\mathbf{C}^{k_1}$ and $\mathbf{C}^{k_2}$, their aperiodic cross-ambiguity function is defined as the sums of the aperiodic cross-ambiguity functions of all subsequences, i.e.,
\begin{equation*}
	\widehat{AF}_{\mathbf{C}^{k_1},\mathbf{C}^{k_2}}(\tau,v)=\sum_{m=0}^{M-1}\widehat{AF}_{\mathbf{c}_{m}^{k_1},\mathbf{c}_{m}^{k_2}}(\tau,v),
\end{equation*}
which is abbreviated to $\widehat{AF}_{\mathbf{C}^{k_1}}(\tau,v)$ when $k_1=k_2$.
For a DRCSS, given a low ambiguity zone (LAZ) $\mathrm{II}\subseteq(-N,N)\times(-N,N)$, the maximum aperiodic ambiguity magnitude of DRCSS $\mathcal{C}$ over this region $\mathrm{II}$ is defined as $\hat{\theta}_{\text{max}}=\text{max}\{\hat{\theta}_a,\hat{\theta}_c\}$, where
$$\hat{\theta}_a=\text{max} \{|\widehat{AF}_{\mathbf{C}}(\tau,v)|: \mathbf{C}\in \mathcal{C}, (\tau,v)\neq (0,0)\in \mathrm{II}\}$$ 
is maximal aperiodic auto-ambiguity magnitude and
$$\hat{\theta}_c=\text{max} \{|\widehat{AF}_{\mathbf{C},\mathbf{D}}(\tau,v)|: \mathbf{C},\mathbf{D}\in \mathcal{C}, (\tau,v)\in \mathrm{II}\}$$    
is maximal aperiodic cross-ambiguity magnitude.

A DRCSS with maximum ambiguity magnitude $\hat{\theta}_{\text{max}}$ over LAZ $\mathrm{II}$ is denoted by $(K,M,N,\hat{\theta}_{\text{max}},\mathrm{II})$-LAZ-DRCSS. For a LAZ-DRCSS, if $\mathrm{II}=(-N,N)\times(-N,N)$, it is recorded as $(K,M,N,\hat{\theta}_{\text{max}})$-DRCSS. Especially, when $M=1$, the DRCSS degenerates into a DRSS \cite{C19} and denoted by $(K,N,\hat{\theta}_{\text{max}},\mathrm{II})$-LAZ-DRSS.

The aperiodic AF magnitude lower bound of DRCSSs can be expressed as follows:

\begin{lemma}\cite{D1}
	For an aperiodic $(K,M,N,\hat{\theta}_{\text{max}},\mathrm{II})$-LAZ-DRCSS, where $\mathrm{II}=(-Z_x,Z_x)\times (-Z_y,Z_y)$ and $1\leq Z_x,Z_y \leq N$, we have 
	\begin{equation}
		\hat{\theta}_{\text{max}} \geq \hat{\theta}_{\text{opt}}= \frac{MN}{\sqrt{Z_y}}\sqrt{\frac{\frac{KZ_xZ_y}{M(N+Z_x-1)}-1}{KZ_x-1}}.
	\end{equation}
\end{lemma}

\begin{lemma}\cite{D2}
	For an aperiodic $(K,M,N,\hat{\theta}_{\text{max}},\mathrm{II})$-LAZ-DRCSS, where $\mathrm{II}=(-Z_x,Z_x)\times (-Z_y,Z_y)$ and $1\leq Z_x,Z_y \leq N$, we have
		\begin{equation}\label{eqD1}
		\hat{\theta}_{\text{max}} \geq \hat{\theta}_{\text{opt}}= \sqrt{MN(1-2\sqrt{\frac{M}{3KZ_y}})},
	\end{equation}
	where $K>\frac{3M}{Z_y}$ and $N\sqrt{\frac{3M}{KZ_y}} \leq Z_x \leq N$.
\end{lemma}

\begin{remark}
	Typically, an optimality factor $\hat{\rho}$ is used to describe the closeness between the maximum AF magnitude of the DRCSS and theoretical lower bound, which is defined by
	\begin{equation}
		\hat{\rho}=\frac{\hat{\theta}_{\text{max}}}{\hat{\theta}_{\text{opt}}}.
	\end{equation}
	Obviously, $\hat{\rho} \geq 1$. If $\hat{\rho}=1$, the DRCSS is said to be optimal. If $\displaystyle\lim_{N\to \infty} \hat{\rho}=1$, the DRCSS is said to be asymptotically optimal.
\end{remark}

\subsection{Peak-to-Average Power Ratio}
Multi-carrier modulation systems suffer from the PAPR problem, which is a key technical challenge. High PAPR can cause saturation in the transmitter front-end circuits, resulting in nonlinear distortion and significantly reducing transmission efficiency. In \cite{R20}, Liu et al. pointed out that effective PAPR mitigation in such systems fundamentally requires to decrease the maximum PAPR of the column sequences within each complementary matrix. For a multicarrier system with $M$ sub-carriers, let $\mathbf{u}=(u_0,u_1,\cdots,u_{M-1})$ be a complex-valued sequence of length $M$ which is spread over $M$ sub-carriers. The time domain multi-carrier signal is defined by

\begin{equation}
	s_\mathbf{u}(t)=
	\begin{cases}
		\sum_{m=0}^{M-1}u_m e^{\sqrt{-1}2\pi mt}, & 0\leq t <1, \\
		0,  & \text{otherwise},
	\end{cases}
\end{equation}
where the carrier spacing has been normalized to $1$. The average power of the multi-carrier signal is $A_\mathbf{u}=M$. Denote the instantaneous power by $P_\mathbf{u}(t)=|s_\mathbf{u}(t)|^2$. The PAPR of the sequence $\mathbf{u}$ under the multi-carrier modulation is defined by 

\begin{equation}
	\text{PAPR}(\mathbf{u})=\max_{0\leq t <1} \frac{P_\mathbf{u}(t)}{A_\mathbf{u}}=\frac{1}{M} \max_{0\leq t <1} |s_\mathbf{u}(t)|^2.
\end{equation}

Besides, the instantaneous-to-average power ratio (IAPR) of $\mathbf{u}$ is defined as 
$$\text{IAPR}_{\mathbf{u}}(t)=\frac{|s_\mathbf{u}(t)|^2}{M}, \;\; 0\leq t <1. $$
Then the PAPR of $\mathbf{u}$ satisfies
$$\text{PAPR}(\mathbf{u})=\max_{0\leq t <1} \text{IAPR}_{\mathbf{u}}(t).$$

In practice, the column sequence PAPR of the complementary matrices in a DRCSS should be taken into consideration. It is desirable that the maximum PAPR of the column sequences for all complementary matrices in a DRCSS is much smaller than $M$.

\subsection{The m-sequence}
Let $q=p^n$, where $p$ is a prime and $n$ is a positive integer. For a positive integer $r$, the trace function from $\mathbb{F}_{q^r}$ to $\mathbb{F}_q$ is defined by
$$\mathrm{Tr}_n^{rn}(x)=x+x^q+\cdots+x^{q^{r-1}}, \;\; x \in \mathbb{F}_q.$$ 
Let $\beta$ be a primitive element of $\mathbb{F}_{q^r}$, the well-known $q$-ary $m$-sequence of period $q^r-1$ is defined by 
$$\mathbf{s}=(s(0),s(1),\cdots,s(q^r-2)),$$
where $s(t)=\mathrm{Tr}_n^{rn}(\beta^t)$ with $t=0,1,\cdots,q^r-2$.

The following lemma is vital to calculate the aperiodic ambiguity magnitude of some families of DRCSSs in this paper.
\begin{lemma}\cite{R21}\label{lemmaD1}
	Let $\mathbf{s}$ be the $m$-sequence defined above. Then every segment of $\frac{q^r-1}{q-1}$ consecutive symbols from $\mathbf{s}$ contains exactly $\frac{q^{r-1}-1}{q-1}$ zeros.
	In particular, when $r=2$, $s(t)=\mathrm{Tr}_n^{2n}(\beta^t)$. In this case, suppose that $s(e)=0$ for $0\leq e \leq q$. Then, $s(t)=0$ if and only if $t \equiv e \pmod{q+1}$.
\end{lemma}

\section{Five new constructions of aperiodic DRCSSs}\label{s3}
The objective of this section to present five new constructions of aperiodic DRCSSs with novel parameters, which are asymptotically optimal with respect to the lower bound in \eqref{eqD1}. In this section, we let $\alpha$ be a primitive element of $\mathbb{F}_q$ and $\beta$ be a primitive element of $\mathbb{F}_{q^2}$.

\subsection{The first construction of aperiodic DRCSSs}
Let $i$ denotes the row index and $j$ denotes the column index for $0\leq i, j < q$, define a matrix $\Psi=[\psi_i^j]$ satisfying the following properties. 
\begin{itemize}
	\item Each entry of $\Psi$ is a $p$-th root of unity, i.e., $\psi_i^j \in \{\xi_p^0, \xi_p^1, \cdots, \xi_p^{p-1}\}$;
	\item The columns of matrix $\Psi$ are pairwise orthogonal.
\end{itemize}

Let $\phi(\cdot)$ be an arbitrary one-to-one mapping from $\mathbb{F}_q$ to $\mathbb{Z}_q$. 
We define a DRCS set
\begin{equation}\label{equation3.1}
	\mathcal{C}=\{\mathbf{C}^{k}: 0\leq k \leq q\},
\end{equation}
where $\mathbf{C}^{k}=[\mathbf{c}_0^{k}, \mathbf{c}_1^{k}, \cdots ,\mathbf{c}_{q-1}^{k}]^{\mathrm{T}}$ consist of $K=q+1$ sequences of length $N=q$ defined by
$$\mathbf{c}_m^{k}=\{c_m^{k}(t)\}_{t=0}^{q-1},$$ 
$$c_m^{k}(t)=\psi_m^{\phi(\mathrm{Tr}_n^{2n}(\beta^{k(q-1)+t}))}, \;\; 0\leq m< q.$$
It is straightforward to show that the alphabet size of $\mathcal{C}$ is $p$.

\begin{theorem}\label{theorem1}
	Let $q=p^n>2$, where $p$ is a prime and $n$ is a positive integer. Then $\mathcal{C}$ is an aperiodic $(q+1,q,q,q)$-DRCSS with alphabet size $p$ which is asymptotically optimal with respect to the lower bound in \eqref{eqD1}.
\end{theorem}

\begin{proof}
	According to the definition of DRCSS, we divide the proof into two cases:
	
	\textbf{Case 1: (auto-AF)} For any $0\leq k \leq q$ , we have
	$$
	\begin{aligned}
		\widehat{AF}_{\mathbf{C}^{k}}(\tau,v)
		&= \sum_{m=0}^{q-1} \widehat{AF}_{\mathbf{c}_m^{k}}(\tau,v) \\
		&= \sum_{m=0}^{q-1} \sum_{t=0}^{q-1-\tau} c_m^{k}(t) (c_m^{k}(t+\tau))^* \xi_{q}^{tv}  \\
		&= \sum_{m=0}^{q-1} \sum_{t=0}^{q-1-\tau} 
		 \psi_m^{\phi(\mathrm{Tr}_n^{2n}(\beta^{k(q-1)+t}))} (\psi_m^{\phi(\mathrm{Tr}_n^{2n}(\beta^{k(q-1)+t+\tau}))})^* \xi_{q}^{tv}\\
		&= \sum_{t=0}^{q-1-\tau} \xi_{q}^{tv} \sum_{m=0}^{q-1}  
		\psi_m^{\phi(\mathrm{Tr}_n^{2n}(\beta^{k(q-1)+t}))} (\psi_m^{\phi(\mathrm{Tr}_n^{2n}(\beta^{k(q-1)+t+\tau}))})^*.
	\end{aligned}
	$$
	
	Since $\phi(\cdot)$ is a one-to-one mapping from $\mathbb{F}_q$ to $\mathbb{Z}_q$, then $\psi_m^{\phi(\mathrm{Tr}_n^{2n}(\beta^{k(q-1)+t}))}=\psi_m^{\phi(\mathrm{Tr}_n^{2n}(\beta^{k(q-1)+t+\tau}))}$ for all $m=0,1,\cdots,q-1$ if and only if $\mathrm{Tr}_n^{2n}(\beta^{k(q-1)+t})=\mathrm{Tr}_n^{2n}(\beta^{k(q-1)+t+\tau})$. Since the columns of matrix $\Psi$ are pairwise orthogonal, thus we only need to consider the $\mathrm{Tr}_n^{2n}(\beta^{k(q-1)+t}(1-\beta^{\tau}))=0$.

	When $\tau=0$, we have 
	$$\widehat{AF}_{\mathbf{C}^{k}}(0,v)=q\sum_{t=0}^{q-1} \xi_{q}^{tv}=\begin{cases}
		q^2, &\text{ if } v=0,\\
		0, &\text{ if } v \neq 0.
	\end{cases}$$
	
	When $\tau \neq 0$, by Lemma \ref{lemmaD1}, we deduce that $\mathrm{Tr}_n^{2n}(\beta^{k(q-1)+t}(1-\beta^{\tau}))=0$ with variable $t$ has at most one solution if $0 \leq t \leq q-1-\tau$.
	
	If $\mathrm{Tr}_n^{2n}(\beta^{k(q-1)+t}(1-\beta^{\tau})) \neq 0$ for $0 \leq t \leq q-1-\tau$, then 
	$$\widehat{AF}_{\mathbf{C}^{k}}(\tau,v)=0.$$
	
	If there exists a unique solution $0\leq t^{\prime} \leq q-1-\tau$ such that  $\mathrm{Tr}_n^{2n}(\beta^{k(q-1)+t^{\prime}}(1-\beta^{\tau})) = 0$, then 
	$$\begin{aligned}
	\widehat{AF}_{\mathbf{C}^{k}}(\tau,v)=&q \xi_q^{t^{\prime}v}+\sum_{t=0,t\neq t^{\prime}}^{q-1-\tau} \xi_q^{tv} \sum_{m=0}^{q-1}\psi_m^{\phi(\mathrm{Tr}_n^{2n}(\beta^{k(q-1)+t}))} (\psi_m^{\phi(\mathrm{Tr}_n^{2n}(\beta^{k(q-1)+t+\tau}))})^*\\
	=&q \xi_q^{t^{\prime}v}.
	\end{aligned}$$

	\textbf{Case 2:(cross-AF)} For any $0 \leq k_1, k_2 \leq q$ and $k_1 \neq k_2$, we have
	$$
	\begin{aligned}
		\widehat{AF}_{\mathbf{C}^{k_1},\mathbf{C}^{k_2}}(\tau,v)
		&= \sum_{m=0}^{q-1} \widehat{AF}_{\mathbf{c}_m^{k_1},\mathbf{c}_m^{k_2}}(\tau,v) \\
		&= \sum_{m=0}^{q-1} \sum_{t=0}^{q-1-\tau} c_m^{k_1}(t) (c_m^{k_2}(t+\tau))^* \xi_{q}^{tv}  \\
		&= \sum_{m=0}^{q-1} \sum_{t=0}^{q-1-\tau} \psi_m^{\phi(\mathrm{Tr}_n^{2n}(\beta^{k_1(q-1)+t}))} (\psi_m^{\phi(\mathrm{Tr}_n^{2n}(\beta^{k_2(q-1)+t+\tau}))})^* \xi_{q}^{tv}\\
		&= \sum_{t=0}^{q-1-\tau} \xi_{q}^{tv} \sum_{m=0}^{q-1} \psi_m^{\phi(\mathrm{Tr}_n^{2n}(\beta^{k_1(q-1)+t}))} (\psi_m^{\phi(\mathrm{Tr}_n^{2n}(\beta^{k_2(q-1)+t+\tau}))})^*.
	\end{aligned}
	$$
	
Since $\phi(\cdot)$ is a one-to-one mapping from $\mathbb{F}_q$ to $\mathbb{Z}_q$, then $\psi_m^{\phi(\mathrm{Tr}_n^{2n}(\beta^{k_1(q-1)+t}))}=\psi_m^{\phi(\mathrm{Tr}_n^{2n}(\beta^{k_2(q-1)+t+\tau}))}$ for all $m=0,1,\cdots,q-1$ if and only if $\mathrm{Tr}_n^{2n}(\beta^{k_1(q-1)+t})=\mathrm{Tr}_n^{2n}(\beta^{k_2(q-1)+t+\tau})$. Since the columns of matrix $\Psi$ are pairwise orthogonal, thus we only need to consider the $\mathrm{Tr}_n^{2n}(\beta^{k_1(q-1)+t}(1-\beta^{(k_2-k_1)(q-1)+\tau}))=0$.

If $\tau=q-1$ and $k_2-k_1=q$, then $t=0$ and $(k_2-k_1)(q-1)+\tau \equiv0 \pmod{q^2-1}$. Accordingly, $1-\beta^{(k_2-k_1)(q-1)+\tau}=0$ and $\psi_m^{\phi(\mathrm{Tr}_n^{2n}(\beta^{k_1(q-1)+t}))}= \psi_m^{\phi(\mathrm{Tr}_n^{2n}(\beta^{k_2(q-1)+t+\tau}))}$ for all $m=0,1,\cdots,q-1$. Consequently,
$$\widehat{AF}_{\mathbf{C}^{k_1},\mathbf{C}^{k_2}}(\tau,v)=q.$$ 

If $(\tau,k_2-k_1)\neq (q-1,q)$, then $(k_2-k_1)(q-1)+\tau \not \equiv0 \pmod{q^2-1}$. By Lemma \ref{lemmaD1}, we deduce that $\mathrm{Tr}_n^{2n}(\beta^{k_1(q-1)+t}(1-\beta^{(k_2-k_1)(q-1)+\tau}))=0$ has at most one solution for $0 \leq t \leq q-1-\tau$. Similarly to Case 1 above, we also have 
$$|\widehat{AF}_{\mathbf{C}^{k_1},\mathbf{C}^{k_2}}(\tau,v)| \in \{0,q\}.$$
	
	Summarizing above cases, we deduce that the maximum aperiodic ambiguity magnitude of $\mathcal{C}$ is $q$. Then $\mathcal{C}$ is an aperiodic $(q+1,q,q,q)$-DRCSS.
	
	According to the lower bound in \eqref{eqD1}, we have $$\hat{\theta}_{\text{opt}}=\sqrt{q^2(1-2\sqrt{\frac{q}{3(q+1)q}})}=q\sqrt{1-2\sqrt{\frac{1}{3(q+1)}}}.$$
	It is easy to see that 
	$$\displaystyle\lim_{q\to \infty} \hat{\rho}
	=\displaystyle\lim_{q\to\infty} \frac{q}{q\sqrt{1-2\sqrt{\frac{1}{3(q+1)}}}}=1.$$
	Thus the aperiodic DRCSS $\mathcal{C}$ is asymptotically optimal.
\end{proof}
			
			\begin{table}[htbp]  
				\centering 
				\caption{The parameters of some aperiodic DRCSSs in Theorem \ref{theorem1} for $q=p$ and $5\leq q \leq 43$} 
				\label{table1} 
				\begin{tabular}{ccccccc} 
					\toprule 
					q     & K & M & N &$\hat{\theta}_\text{max}$  & $\hat{\theta}_{\text{opt}}$ & $\hat{\rho}$   \\
					\midrule 
					5     & 6 & 5 & 5 &5         & 3.6352 & 1.3754 \\
					7     & 8 & 7 & 7 &7         & 5.3848 & 1.3000 \\
					11    & 12 & 11 & 11 &11     & 8.9815 & 1.2247 \\
					13    & 14 & 13 & 13 &13     & 10.8095 & 1.2026\\
					17    & 18 & 17 & 17 &17     & 14.5032 & 1.1722 \\
					19    & 20 & 19 & 19 &19     & 16.3643 & 1.1611\\
					23    & 24 & 23 & 23 &23     & 20.1075 & 1.1438 \\
					29    & 30 & 29 & 29 &29     & 25.7624 & 1.1257 \\
					31    & 32 & 31 & 31 &31     & 27.6557 & 1.1209 \\
					37    & 38 & 37 & 37 &37     & 33.3557 & 1.1093\\
					41    & 42 & 41 & 41 &41     & 37.1684 & 1.1031 \\
					43    & 44 & 43 & 43 &43     & 39.0785 & 1.1003 \\
					\bottomrule 
				\end{tabular}
			\end{table}
			
		\begin{remark}
			In Table \ref{table1}, we list the parameters of some aperiodic DRCSSs $\mathcal{C}$ constructed in Theorem \ref{theorem1}. The numerical date show that the optimal factor of $\mathcal{C}$ approaches $1$ quickly when $q$ increases.
		\end{remark}
		
		\begin{table}[htbp]
			\centering
			\caption{The DRCS Set $\mathcal{C}$ in Example \ref{example1}}
			\label{table2}
			\begin{tabular}{|c|c|c|}
				\hline
				$\mathbf{C}^{0}$ & $\mathbf{C}^{1}$  & $\mathbf{C}^{2}$ \\
				\hline
				$\begin{bmatrix}          
					0 0 0 0 0  \\
					2 3 2 0 1 \\
					4 1 4 0 2 \\
					3 2 3 0 4 \\
					1 4 1 0 3 
				\end{bmatrix}$
				&
				$\begin{bmatrix}
				0 0 0 0 0  \\
				1 3 3 4 3 \\
				2 1 1 3 1 \\
				4 2 2 1 2 \\
				3 4 4 2 4 
				\end{bmatrix}$
				&
				$\begin{bmatrix}
					0 0 0 0 0  \\
					3 0 2 4 4 \\
					1 0 4 3 3 \\
					2 0 3 1 1 \\
					4 0 1 2 2 
				\end{bmatrix}$ \\
				\hline
				$\mathbf{C}^{3}$ & $\mathbf{C}^{4}$ & $\mathbf{C}^{5}$ \\
				\hline
				$\begin{bmatrix}
				0 0 0 0 0  \\
				2 1 2 0 3 \\
				4 2 4 0 1 \\
				3 4 3 0 2 \\
				1 3 1 0 4 
				\end{bmatrix}$
				&
				$\begin{bmatrix}
				0 0 0 0 0  \\
				3 1 1 2 1 \\
				1 2 2 4 2 \\
				2 4 4 3 4 \\
				4 3 3 1 3 
				\end{bmatrix}$
				&
				$\begin{bmatrix}
			    0 0 0 0 0  \\
				1 0 4 2 2 \\
				2 0 3 4 4 \\
				4 0 1 3 3 \\
				3 0 2 1 1 
				\end{bmatrix}$ \\
				\hline
			\end{tabular}
		\end{table}
		
		\begin{figure}[!htbp]  
			\centering
			\includegraphics[width=0.9\linewidth]{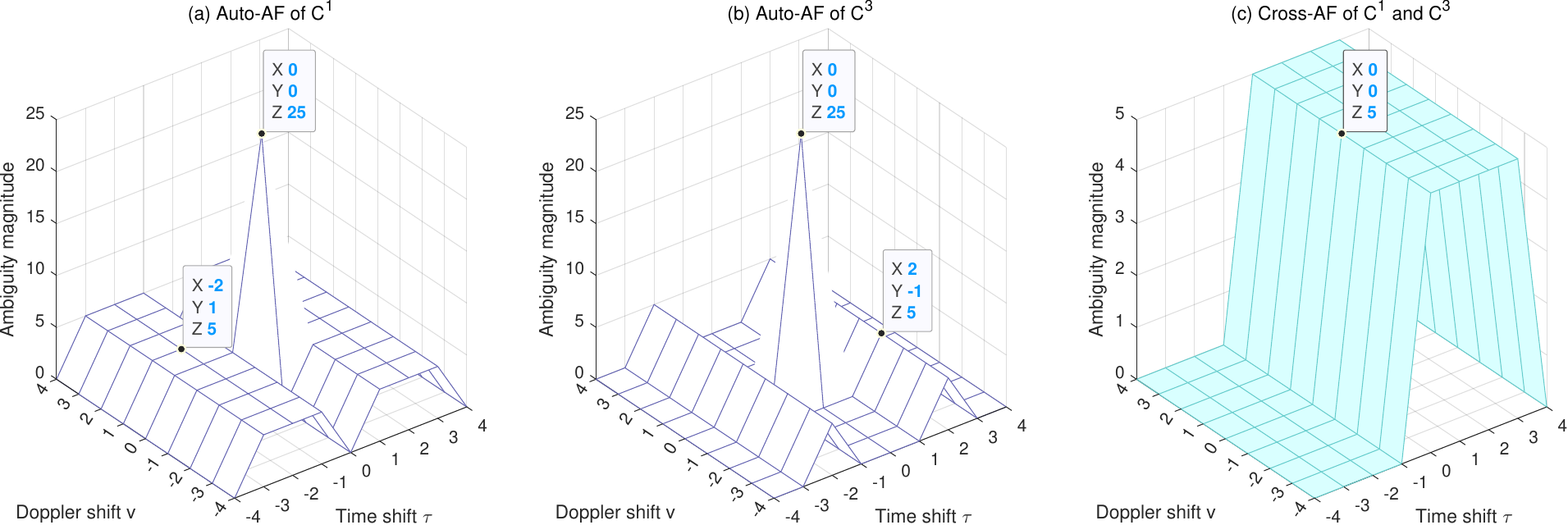}
			\caption{The aperiodic auto-ambiguity and cross-ambiguity magnitude of $\mathbf{C}^{1}$ and $\mathbf{C}^{3}$ in Example \ref{example1}}
			\label{Fig1}  
		\end{figure}
		
		\begin{example}\label{example1}
			Let $q=5$ and $\beta$ be a primitive element of $\mathbb{F}_{5^2}$ satisfying $\beta^2+\beta+2=0$. Consider a one-to-one mapping function from $\mathbb{F}_5$ to $\mathbb{Z}_5$, defined by $\phi(x)=x$ for $x \in \mathbb{F}_5$. Choose  
			$$\Psi=\begin{bmatrix}
				 0 0 0 0 0\\
				 0 1 2 4 3\\
				0 2 4 3 1\\
				0 4 3 1 2\\
				0 3 1 2 4
			\end{bmatrix},$$
			where each entry stands for a power of $\xi_5$. Then $\mathcal{C}$ is a $(6,5,5,5)$-DRCSS. According to Theorem \ref{theorem1}, we obtain a DRCSS $\mathcal{C}$ are known in Table \ref{table2}, where each entry represents a power of $\xi_5$. By MATLAB problem, we respectively show the auto-ambiguity magnitude distribution of $\mathbf{C}^1$, the auto-ambiguity magnitude distribution of $\mathbf{C}^3$ and the cross-ambiguity magnitude distribution of $\mathbf{C}^1$ and $\mathbf{C}^3$ in Fig. \ref{Fig1}.
		\end{example}
		
		\begin{remark}\label{remark}
			We could derive an aperiodic DRCSSs $\mathcal{C}$ in Theorem \ref{theorem1} with low PAPR. Choose any column sequence of $\mathbf{C}^k$ in Theorem \ref{theorem1} as
			$$\mathbf{c}(t)=(\psi_0^{\phi(\mathrm{Tr}_n^{2n}(\beta^{k(q-1)+t}))},\psi_1^{\phi(\mathrm{Tr}_n^{2n}(\beta^{k(q-1)+t}))},\cdots,\psi_{q-1}^{\phi(\mathrm{Tr}_n^{2n}(\beta^{k(q-1)+t}))})^{\mathrm{T}}.$$
			Hence it sufficient to find a suitable column orthogonal matrix $\Psi$ with low PAPR. In fact, there exist such matrices in the literature. For the case of $p=2$, the construction of column orthogonal matrices with PAPR bounded by $2$ has been extensively studied, and the interested reader is referred to \cite{q1,q2,q3,q4} for detailed constructions and theoretical proofs. For $p \geq 3$, matrices satisfying the column orthogonality and PAPR constraint of at most $p$ can be found in \cite{q5,q6}.
		\end{remark}
		
		\subsection{The second construction of aperiodic DRCSSs}
		Let $\phi(\cdot)$ be an arbitrary one-to-one mapping from $\mathbb{F}_q$ to $\mathbb{Z}_q$. 
		Let $\Psi=[\psi_i^j]_{q\times q}$ be the column orthogonal matrix defined as above, where $i$ denotes the row index and $j$ denotes the column index for $0 \leq i,j \leq q-1$.
		We define a DRCS set
		\begin{equation}\label{equation3.2}
			\mathcal{C}=\{\mathbf{C}^{k}: 0\leq k \leq q\},
		\end{equation}
		where $\mathbf{C}^{k}=[\mathbf{c}_0^{k}, \mathbf{c}_1^{k}, \cdots ,\mathbf{c}_{q-1}^{k}]^{\mathrm{T}}$ consist of $K=q+1$ sequences of length $N=q-1$ defined by
		$$\mathbf{c}_m^{k}=\{c_m^{k}(t)\}_{t=0}^{q-2},$$ 
		$$c_m^{k}(t)=\psi_m^{\phi(\mathrm{Tr}_n^{2n}(\beta^{k(q-1)+t}))}, \;\; 0\leq m< q.$$
		It is straightforward to show that the alphabet size of $\mathcal{C}$ is $p$.

		\begin{theorem}\label{theorem2}
			Let $q=p^n>3$, where $p$ is a prime and $n$ is a positive integer. Then $\mathcal{C}$ is an aperiodic $(q+1,q,q-1,q)$-DRCSS with alphabet size $p$ which is asymptotically optimal with respect to the lower bound in \eqref{eqD1}.
		\end{theorem}
		
		\begin{proof}
			The proof is very similar to that Theorem \ref{theorem1}. We omit it here.
		\end{proof}
		
		\begin{table}[htbp]  
			\centering 
			\caption{The parameters of some aperiodic DRCSSs in Theorem \ref{theorem2} for $q=p$ and $5\leq q \leq 43$} 
			\label{table3} 
			\begin{tabular}{ccccccc} 
				\toprule 
				q     & K & M & N &$\hat{\theta}_\text{max}$  & $\hat{\theta}_{\text{opt}}$ & $\hat{\rho}$   \\
				\midrule 
				5     & 6 & 5 & 4 &5         & 3.0756 & 1.6257 \\
				7     & 8 & 7 & 6 &7         & 4.8456 & 1.4446 \\
				11    & 12 & 11 & 10 &11     & 8.4583 & 1.3005 \\
				13    & 14 & 13 & 12 &13     & 10.2904 & 1.2633\\
				17    & 18 & 17 & 16 &17     & 13.9890 & 1.2152 \\
				19    & 20 & 19 & 18 &19     & 15.8517 & 1.1986\\
				23    & 24 & 23 & 22 &23     & 19.5973 & 1.1736 \\
				29    & 30 & 29 & 28 &29     & 25.2544 & 1.1483 \\
				31    & 32 & 31 & 30 &31     & 27.1482 & 1.1419 \\
				37    & 38 & 37 & 36 &37     & 32.8489 & 1.1264\\
				41    & 42 & 41 & 40 &41     & 36.6628 & 1.1183 \\
				43    & 44 & 43 & 42 &43     & 38.5732 & 1.1147 \\
				\bottomrule 
			\end{tabular}
		\end{table}
		
		\begin{remark}
			In Table \ref{table3}, we list the parameters of some aperiodic DRCSSs $\mathcal{C}$ constructed in Theorem \ref{theorem2}. The numerical date show that the optimal factor of $\mathcal{C}$ approaches $1$ quickly when $q$ increases.
		\end{remark}
		
		\begin{remark}
			Note that the DRCSS in Theorem \ref{theorem2} also depends on the column orthogonal matrix $\Psi$. Similarly to the discussions in Remark \ref{remark}, there also exist known suitable matrices $\Psi$ such that the column sequence PAPR of each complementary matrix in the aperiodic DRCSSs $\mathcal{C}$ is not larger than $p$. We omit the details here.
		\end{remark}

		\begin{table}[htbp]
			\centering
			\caption{The DRCS Set $\mathcal{C}$ in Example \ref{example2}}
			\label{table4}
			\begin{tabular}{|c|c|c|}
				\hline
				$\mathbf{C}^{0}$ & $\mathbf{C}^{1}$  & $\mathbf{C}^{2}$ \\
				\hline
				$\begin{bmatrix}          
					0 0 0 0   \\
					2 3 2 0  \\
					4 1 4 0 \\
					3 2 3 0 \\
					1 4 1 0  
				\end{bmatrix}$
				&
				$\begin{bmatrix}
					0 0 0 0   \\
					1 3 3 4  \\
					2 1 1 3  \\
					4 2 2 1  \\
					3 4 4 2  
				\end{bmatrix}$
				&
				$\begin{bmatrix}
					0 0 0 0  \\
					3 0 2 4  \\
					1 0 4 3 \\
					2 0 3 1  \\
					4 0 1 2  
				\end{bmatrix}$ \\
				\hline
				$\mathbf{C}^{3}$ & $\mathbf{C}^{4}$ & $\mathbf{C}^{5}$ \\
				\hline
				$\begin{bmatrix}
					0 0 0 0   \\
					2 1 2 0  \\
					4 2 4 0  \\
					3 4 3 0  \\
					1 3 1 0  
				\end{bmatrix}$
				&
				$\begin{bmatrix}
					0 0 0 0  \\
					3 1 1 2  \\
					1 2 2 4 \\
					2 4 4 3  \\
					4 3 3 1  
				\end{bmatrix}$
				&
				$\begin{bmatrix}
					0 0 0 0   \\
					1 0 4 2  \\
					2 0 3 4  \\
					4 0 1 3  \\
					3 0 2 1  
				\end{bmatrix}$ \\
				\hline
			\end{tabular}
		\end{table}
		
		\begin{figure}[!htbp]  
			\centering
			\includegraphics[width=0.9\linewidth]{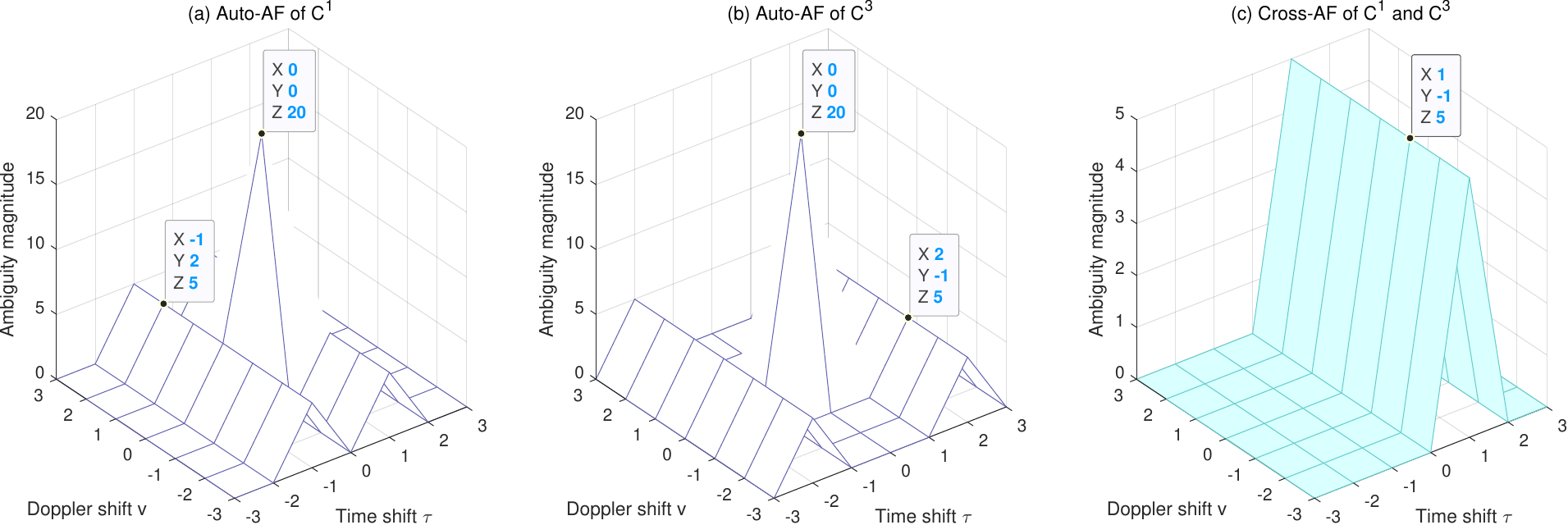}
			\caption{The aperiodic auto-ambiguity and cross-ambiguity magnitude of $\mathbf{C}^{1}$ and $\mathbf{C}^{3}$ in Example \ref{example2}}
			\label{Fig2}  
		\end{figure}
		
		\begin{example}\label{example2}
			Let $q=5$ and $\beta$ be a primitive element of $\mathbb{F}_{5^2}$ satisfying $\beta^2+\beta+2=0$. Consider a one-to-one mapping function from $\mathbb{F}_5$ to $\mathbb{Z}_5$, defined by $\phi(x)=x$ for $x \in \mathbb{F}_5$. Choose  
			$$\Psi=\begin{bmatrix}
				0 0 0 0 0\\
				0 1 2 4 3\\
				0 2 4 3 1\\
				0 4 3 1 2\\
				0 3 1 2 4
			\end{bmatrix},$$
			where each entry stands for a power of $\xi_5$. Then $\mathcal{C}$ is a $(6,5,4,5)$-DRCSS. According to Theorem \ref{theorem2}, we obtain a DRCSS $\mathcal{C}$ are known in Table \ref{table4}, where each entry represents a power of $\xi_5$. By MATLAB problem, we respectively show the auto-ambiguity magnitude distribution of $\mathbf{C}^1$, the auto-ambiguity magnitude distribution of $\mathbf{C}^3$ and the cross-ambiguity magnitude distribution of $\mathbf{C}^1$ and $\mathbf{C}^3$ in Fig. \ref{Fig2}.
		\end{example}

		\subsection{The third construction of aperiodic DRCSSs}
		Let $\phi(\cdot)$ be an arbitrary one-to-one mapping from $\mathbb{F}_q$ to $\mathbb{Z}_q$. 
		Let $\Psi=[\psi_i^j]_{q\times q}$ be the column orthogonal matrix defined as above, where $i$ denotes the row index and $j$ denotes the column index for $0 \leq i,j \leq q-1$.
		We define a DRCS set
		\begin{equation}\label{equation3.3}
			\mathcal{C}=\{\mathbf{C}^{k}: 0\leq k \leq q-2\},
		\end{equation}
		where $\mathbf{C}^{k}=[\mathbf{c}_0^{k}, \mathbf{c}_1^{k}, \cdots ,\mathbf{c}_{q-1}^{k}]^{\mathrm{T}}$ consist of $K=q-1$ sequences of length $N=q+1$ defined by
		$$\mathbf{c}_m^{k}=\{c_m^{k}(t)\}_{t=0}^{q},$$ 
		$$c_m^{k}(t)=\psi_m^{\phi(\mathrm{Tr}_n^{2n}(\beta^{k(q+1)+t}))}, \;\; 0\leq m< q.$$
		It is straightforward to show that the alphabet size of $\mathcal{C}$ is $p$.

		\begin{theorem}\label{theorem3}
			Let $q=p^n>3$, where $p$ is a prime and $n$ is a positive integer. Then $\mathcal{C}$ is an aperiodic $(q-1,q,q+1,q)$-DRCSS with alphabet size $p$ which is asymptotically optimal with respect to the lower bound in \eqref{eqD1}.
		\end{theorem}
		
		\begin{proof}
		The proof is very similar to that Theorem \ref{theorem1}. We omit it here.
		\end{proof}
		
		\begin{table}[htbp]  
			\centering 
			\caption{The parameters of some aperiodic DRCSSs in Theorem \ref{theorem3} for $q=p$ and $5\leq q \leq 43$} 
			\label{table5} 
			\begin{tabular}{ccccccc} 
				\toprule 
				q     & K & M & N &$\hat{\theta}_\text{max}$  & $\hat{\theta}_{\text{opt}}$ & $\hat{\rho}$   \\
				\midrule 
				5     & 4 & 5 & 6 &5         & 3.7668 & 1.3274 \\
				7     & 6 & 7 & 8 &7         & 5.5952 & 1.2511 \\
				11    & 10 & 11 & 12 &11     & 9.2657 & 1.1872 \\
				13    & 12 & 13 & 14 &13     & 11.1149 & 1.1696\\
				17    & 16 & 17 & 18 &17     & 14.8376 & 1.1457 \\
				19    & 18 & 19 & 20 &19     & 16.7092 & 1.1371\\
				23    & 22 & 23 & 24 &23     & 20.4687 & 1.1237 \\
				29    & 28 & 29 & 30 &29     & 26.1408 & 1.1094 \\
				31    & 30 & 31 & 32 &31     & 28.0386 & 1.1056 \\
				37    & 36 & 37 & 38 &37     & 33.7491 & 1.0963\\
				41    & 40 & 41 & 42 &41     & 37.5682 & 1.0913 \\
				43    & 42 & 43 & 44 &43     & 39.4810 & 1.0891 \\
				\bottomrule 
			\end{tabular}
		\end{table}
		
		\begin{remark}
			In Table \ref{table5}, we list the parameters of some aperiodic DRCSSs $\mathcal{C}$ constructed in Theorem \ref{theorem3}. The numerical date show that the optimal factor of $\mathcal{C}$ approaches $1$ quickly when $q$ increases.
		\end{remark}
		
		\begin{remark}
			Note that the DRCSS in Theorem \ref{theorem3} also depends on the column orthogonal matrix $\Psi$. Similarly to the discussions in Remark \ref{remark}, there also exist known suitable matrices $\Psi$ such that the column sequence PAPR of each complementary matrix in the aperiodic DRCSSs $\mathcal{C}$ is not larger than $p$. We omit the details here.
		\end{remark}

		\begin{table}[htbp]
			\centering
			\caption{The DRCS Set $\mathcal{C}$ in Example \ref{example3}}
			\label{table6}
			\begin{tabular}{|c|c|c|c|}
				\hline
				$\mathbf{C}^{0}$ & $\mathbf{C}^{1}$  & $\mathbf{C}^{2}$ & $\mathbf{C}^{3}$ \\
				\hline
				$\begin{bmatrix}          
					0 0 0 0 0 0  \\
					2 3 2 0 1 3 \\
					4 1 4 0 2 1 \\
					3 2 3 0 4 2 \\
					1 4 1 0 3 4
				\end{bmatrix}$
				&
				$\begin{bmatrix}
					0 0 0 0 0 0  \\
					3 4 3 0 2 4 \\
					1 3 1 0 4 3\\
					2 1 2 0 3 1\\
					4 2 4 0 1 2
				\end{bmatrix}$
				&
				$\begin{bmatrix}
					0 0 0 0 0 0  \\
					4 1 4 0 3 1 \\
					3 2 3 0 1 2\\
					1 4 1 0 2 4\\
					2 3 2 0 4 3
				\end{bmatrix}$ 
				&
				$\begin{bmatrix}
					0 0 0 0 0 0  \\
					1 2 1 0 4 2 \\
					2 4 2 0 3 4\\
					4 3 4 0 1 3\\
					3 1 3 0 2 1 
				\end{bmatrix}$\\
		\hline
			\end{tabular}
		\end{table}
		
		\begin{figure}[!htbp]  
			\centering
			\includegraphics[width=0.9\linewidth]{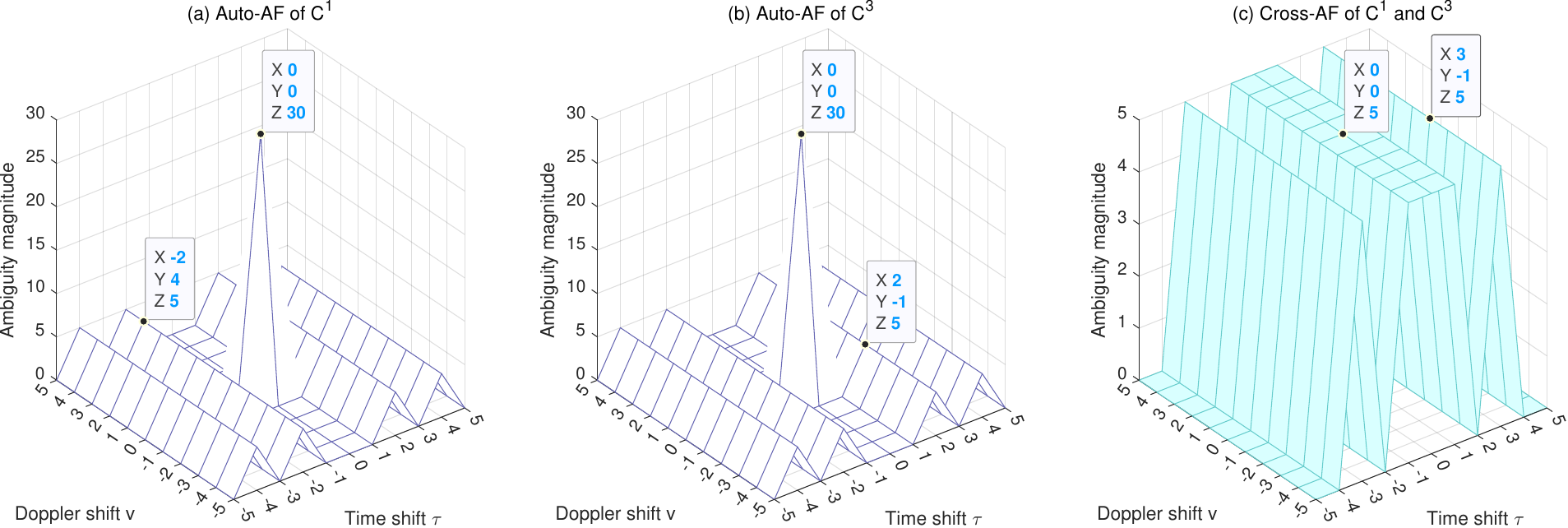}
			\caption{The aperiodic auto-ambiguity and cross-ambiguity magnitude of $\mathbf{C}^{1}$ and $\mathbf{C}^{3}$ in Example \ref{example3}}
			\label{Fig3}  
		\end{figure}
		
		\begin{example}\label{example3}
			Let $q=5$ and $\beta$ be a primitive element of $\mathbb{F}_{5^2}$ satisfying $\beta^2+\beta+2=0$. Consider a one-to-one mapping function from $\mathbb{F}_5$ to $\mathbb{Z}_5$, defined by $\phi(x)=x$ for $x \in \mathbb{F}_5$. Choose  
			$$\Psi=\begin{bmatrix}
				0 0 0 0 0\\
				0 1 2 4 3\\
				0 2 4 3 1\\
				0 4 3 1 2\\
				0 3 1 2 4
			\end{bmatrix},$$
			where each entry stands for a power of $\xi_5$. Then $\mathcal{C}$ is a $(4,5,6,5)$-DRCSS. According to Theorem \ref{theorem3}, we obtain a DRCSS $\mathcal{C}$ are known in Table \ref{table6}, where each entry represents a power of $\xi_5$. By MATLAB problem, we respectively show the auto-ambiguity magnitude distribution of $\mathbf{C}^1$, the auto-ambiguity magnitude distribution of $\mathbf{C}^3$ and the cross-ambiguity magnitude distribution of $\mathbf{C}^1$ and $\mathbf{C}^3$ in Fig. \ref{Fig3}.
		\end{example}

		\subsection{The fourth construction of aperiodic DRCSSs}
		Let $0 \leq e \leq q$ be an integer satisfying $\mathrm{Tr}_n^{2n}(\beta^e)=0$, where $\beta$ is a primitive element of $\mathbb{F}_{q^2}$.
		Let $\phi(\cdot)$ be an arbitrary one-to-one mapping from $\mathbb{F}_q$ to $\mathbb{Z}_q$. 
		We define a DRCS set
		\begin{equation}\label{equation3.4}
			\mathcal{C}=\{\mathbf{C}^{k}: 0\leq k \leq q-2\},
		\end{equation}
		where $\mathbf{C}^{k}=[\mathbf{c}_0^{k}, \mathbf{c}_1^{k}, \cdots ,\mathbf{c}_{q-1}^{k}]^{\mathrm{T}}$ consist of $K=q-1$ sequences of length $N=q$ defined by
		$$\mathbf{c}_m^{k}=\{c_m^{k}(t)\}_{t=0}^{q-1},$$ 
		$$c_m^{k}(t)=\xi_{q-1}^{m\phi(\mathrm{Tr}_n^{2n}(\beta^{e+k(q+1)+t+1}))}, \;\; 0\leq m\leq q-2.$$
		It is straightforward to show that the alphabet size of $\mathcal{C}$ is $q-1$.

		\begin{theorem}\label{theorem4}
			Let $q=p^n>3$, where $p$ is a prime and $n$ is a positive integer. Then $\mathcal{C}$ is an aperiodic $(q-1,q-1,q,q-1)$-DRCSS with alphabet size $q-1$ which is asymptotically optimal with respect to the lower bound in \eqref{eqD1}.
		\end{theorem}
		
		\begin{proof}
			According to the definition of DRCSS, we divide the proof into two cases:
			
			\textbf{Case 1: (auto-AF)} For any $0\leq k \leq q-2$ , we have
			$$
			\begin{aligned}
				\widehat{AF}_{\mathbf{C}^{k}}(\tau,v)
				&= \sum_{m=0}^{q-2} \widehat{AF}_{\mathbf{c}_m^{k}}(\tau,v) \\
				&= \sum_{m=0}^{q-2} \sum_{t=0}^{q-1-\tau} c_m^{k}(t) (c_m^{k}(t+\tau))^* \xi_{q}^{tv}  \\
				&= \sum_{m=0}^{q-2} \sum_{t=0}^{q-1-\tau} 
				\xi_{q-1}^{m\phi(\mathrm{Tr}_n^{2n}(\beta^{e+k(q+1)+t+1}))} (\xi_{q-1}^{m\phi(\mathrm{Tr}_n^{2n}(\beta^{e+k(q+1)+t+1+\tau}))})^* \xi_{q}^{tv}\\
				&= \sum_{t=0}^{q-1-\tau} \xi_{q}^{tv} \sum_{m=0}^{q-2}  
				\xi_{q-1}^{m(\phi(\mathrm{Tr}_n^{2n}(\beta^{e+k(q+1)+t+1}))-\phi(\mathrm{Tr}_n^{2n}(\beta^{e+k(q+1)+t+1+\tau})))} .
			\end{aligned}
			$$
			
			Since $\phi(\cdot)$ is a one-to-one mapping from $\mathbb{F}_q$ to $\mathbb{Z}_q$, then $\phi(\mathrm{Tr}_n^{2n}(\beta^{e+k(q+1)+t+1}))=\phi(\mathrm{Tr}_n^{2n}(\beta^{e+k(q+1)+t+1+\tau}))$ if and only if $\mathrm{Tr}_n^{2n}(\beta^{e+k(q+1)+t+1})=\mathrm{Tr}_n^{2n}(\beta^{e+k(q+1)+t+1+\tau})$. Thus we only need to consider the $\mathrm{Tr}_n^{2n}(\beta^{e+k(q+1)+t+1}(1-\beta^{\tau}))=0$.

			When $\tau=0$, we have 
			$$\widehat{AF}_{\mathbf{C}^{k}}(0,v)=(q-1)\sum_{t=0}^{q-1} \xi_{q}^{tv}=\begin{cases}
				q^2-q, &\text{ if } v=0,\\
				0, &\text{ if } v \neq 0.
			\end{cases}$$
			
			When $\tau \neq 0$, by Lemma \ref{lemmaD1}, we deduce that $\mathrm{Tr}_n^{2n}(\beta^{e+k(q+1)+t+1}(1-\beta^{\tau}))=0$ with variable $t$ has at most one solution if $0 \leq t \leq q-1-\tau$.
			
			If $\mathrm{Tr}_n^{2n}(\beta^{e+k(q+1)+t+1}(1-\beta^{\tau})) \neq 0$ for $0 \leq t \leq q-1-\tau$, then 
			$$\widehat{AF}_{\mathbf{C}^{k}}(\tau,v)=0.$$
			
			If there exists a unique solution $0\leq t^{\prime} \leq q-1-\tau$ such that  $\mathrm{Tr}_n^{2n}(\beta^{e+k(q+1)+t^{\prime}+1}(1-\beta^{\tau})) = 0$, then 
			$$\begin{aligned}
				\widehat{AF}_{\mathbf{C}^{k}}(\tau,v)=&(q-1) \xi_q^{t^{\prime}v}+\sum_{t=0,t\neq t^{\prime}}^{q-1-\tau}\xi_q^{tv} \sum_{m=0}^{q-2} \xi_{q-1}^{m(\phi(\mathrm{Tr}_n^{2n}(\beta^{e+k(q+1)+t+1}))-\phi(\mathrm{Tr}_n^{2n}(\beta^{k(q+1)+t+\tau})))} \\
				=&(q-1) \xi_q^{t^{\prime}v}.
			\end{aligned}$$

			\textbf{Case 2:(cross-AF)} For any $0 \leq k_1, k_2 \leq q-2$ and $k_1 \neq k_2$, we have
			$$
			\begin{aligned}
				\widehat{AF}_{\mathbf{C}^{k_1},\mathbf{C}^{k_2}}(\tau,v)
				&= \sum_{m=0}^{q-2} \widehat{AF}_{\mathbf{c}_m^{k_1},\mathbf{c}_m^{k_2}}(\tau,v) \\
				&= \sum_{m=0}^{q-2} \sum_{t=0}^{q-1-\tau} c_m^{k_1}(t) (c_m^{k_2}(t+\tau))^* \xi_{q}^{tv}  \\
				&= \sum_{m=0}^{q-2} \sum_{t=0}^{q-1-\tau} \xi_{q-1}^{m\phi(\mathrm{Tr}_n^{2n}(\beta^{e+k_1(q+1)+t+1}))} (\xi_{q-1}^{m\phi(\mathrm{Tr}_n^{2n}(\beta^{e+k_2(q+1)+t+1+\tau}))})^* \xi_{q}^{tv}\\
				&= \sum_{t=0}^{q-1-\tau} \xi_{q}^{tv} \sum_{m=0}^{q-2} \xi_{q-1}^{m(\phi(\mathrm{Tr}_n^{2n}(\beta^{k_1(q+1)+t}))-\phi(\mathrm{Tr}_n^{2n}(\beta^{k_2(q+1)+t+\tau})))}.
			\end{aligned}
			$$
			
			Since $\phi(\cdot)$ is a one-to-one mapping from $\mathbb{F}_q$ to $\mathbb{Z}_q$, then $\phi(\mathrm{Tr}_n^{2n}(\beta^{e+k_1(q+1)+t+1}))=\phi(\mathrm{Tr}_n^{2n}(\beta^{e+k_2(q+1)+t+1+\tau}))$ if and only if $\mathrm{Tr}_n^{2n}(\beta^{e+k_1(q+1)+t+1})=\mathrm{Tr}_n^{2n}(\beta^{e+k_2(q+1)+t+1+\tau})$. Thus we only need to consider the $\mathrm{Tr}_n^{2n}(\beta^{e+k_1(q+1)+t+1}(1-\beta^{(k_2-k_1)(q+1)+\tau}))=0$.
			
			It is clear that $(k_2-k_1)(q+1)+\tau \not \equiv0 \pmod{q^2-1}$. By Lemma \ref{lemmaD1}, we deduce that $\mathrm{Tr}_n^{2n}(\beta^{e+k_1(q+1)+t+1}(1-\beta^{(k_2-k_1)(q+1)+\tau}))=0$ has at most one solution for $0 \leq t \leq q-1-\tau$. Similarly to Case 1 above, we also have 
			$$|\widehat{AF}_{\mathbf{C}^{k_1},\mathbf{C}^{k_2}}(\tau,v)| \in \{0,q-1\}.$$
			
			Summarizing above cases, we deduce that the maximum aperiodic ambiguity magnitude of $\mathcal{C}$ is $q-1$. Then $\mathcal{C}$ is an aperiodic $(q-1,q-1,q,q-1)$-DRCSS.
			
			According to the lower bound in \eqref{eqD1}, we have $$\hat{\theta}_{\text{opt}}=\sqrt{q(q-1)(1-2\sqrt{\frac{q-1}{3(q-1)q}})}=\sqrt{q(q-1)}\sqrt{1-2\sqrt{\frac{1}{3q}}}.$$
			It is easy to see that 
			$$\displaystyle\lim_{q\to \infty} \hat{\rho}
			=\displaystyle\lim_{q\to\infty} \frac{q-1}{\sqrt{q(q-1)}\sqrt{1-2\sqrt{\frac{1}{3q}}}}=1.$$
			Thus the aperiodic DRCSS $\mathcal{C}$ is asymptotically optimal.
		\end{proof}
		
		\begin{table}[htbp]  
			\centering 
			\caption{The parameters of some aperiodic DRCSSs in Theorem \ref{theorem4} for $q=p$ and $5\leq q \leq 43$} 
			\label{table7} 
			\begin{tabular}{ccccccc} 
				\toprule 
				q     & K & M & N &$\hat{\theta}_\text{max}$  & $\hat{\theta}_{\text{opt}}$ & $\hat{\rho}$   \\
				\midrule 
		5     & 4 & 4 & 5 &4         & 3.1100 & 1.2862 \\
		7     & 6 & 6 & 7 &6         & 4.8651 & 1.2332 \\
		11    & 10 & 10 & 11 &10     & 8.4678 & 1.1810 \\
		13    & 12 & 12 & 13 &12     & 10.2976 & 1.1653\\
		17    & 16 & 16 & 17 &16     & 13.9937 & 1.1433 \\
		19    & 18 & 18 & 19 &18     & 15.8557 & 1.1352\\
		23    & 22 & 22 & 23 &22     & 19.6002 & 1.1224 \\
		29    & 28 & 28 & 29 &28     & 25.2565 & 1.1086 \\
		31    & 30 & 30 & 31 &30     & 27.1501 & 1.1050 \\
		37    & 36 & 36 & 37 &36     & 32.8503 & 1.0959\\
		41    & 40 & 40 & 41 &40     & 36.6640 & 1.0910 \\
		43    & 42 & 42 & 43 &42     & 38.5744 & 1.0888 \\
				\bottomrule 
			\end{tabular}
		\end{table}
		
		\begin{remark}
			In Table \ref{table7}, we list the parameters of some aperiodic DRCSSs $\mathcal{C}$ constructed in Theorem \ref{theorem4}. The numerical date show that the optimal factor of $\mathcal{C}$ approaches $1$ quickly when $q$ increases.
		\end{remark}

		\begin{table}[htbp]
			\centering
			\caption{The DRCS Set $\mathcal{C}$ in Example \ref{example4}}
			\label{table8}
			\begin{tabular}{|c|c|c|c|}
				\hline
				$\mathbf{C}^{0}$ & $\mathbf{C}^{1}$  & $\mathbf{C}^{2}$ & $\mathbf{C}^{3}$ \\
				\hline
				$\begin{bmatrix}          
					0 0 0 0 0   \\
					1 0 0 3 0  \\
					2 0 0 2 0  \\				
					3 0 0 1 0 
				\end{bmatrix}$
				&
				$\begin{bmatrix}
					0 0 0 0 0   \\
					2 3 3 1 3  \\
					0 2 2 2 2 \\
					2 1 1 3 1 
				\end{bmatrix}$
				&
				$\begin{bmatrix}
					0 0 0 0 0  \\
					0 1 1 2 1  \\
					0 2 2 0 2 \\
					0 3 3 2 3 
				\end{bmatrix}$ 
				&
				$\begin{bmatrix}
					0 0 0 0 0   \\
					3 2 2 0 2  \\
					2 0 0 0 0 \\
					1 2 2 0 2  
				\end{bmatrix}$\\
				\hline
			\end{tabular}
		\end{table}
		
		\begin{figure}[!htbp]  
			\centering
			\includegraphics[width=0.9\linewidth]{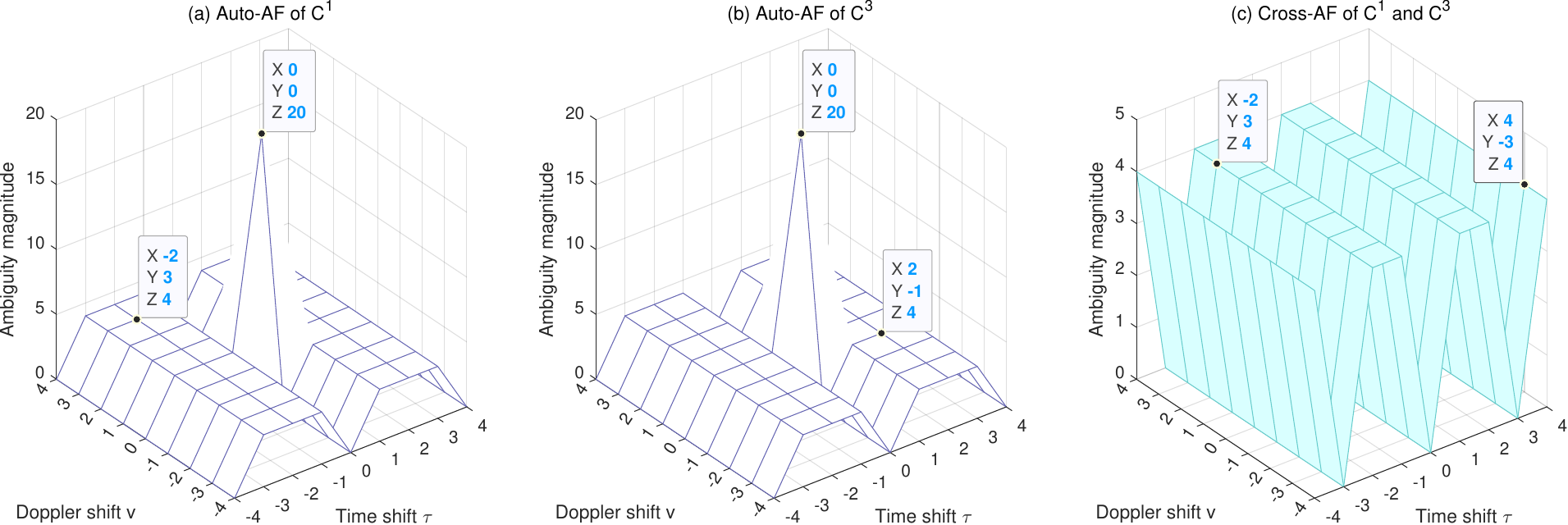}
			\caption{The aperiodic auto-ambiguity and cross-ambiguity magnitude of $\mathbf{C}^{1}$ and $\mathbf{C}^{3}$ in Example \ref{example4}}
			\label{Fig4}  
		\end{figure}
		
		\begin{example}\label{example4}
			Let $q=5$ and $\beta$ be a primitive element of $\mathbb{F}_{5^2}$ satisfying $\beta^2+\beta+2=0$. Consider a one-to-one mapping function from $\mathbb{F}_5$ to $\mathbb{Z}_5$, defined by $\phi(x)=x$ for $x \in \mathbb{F}_5$. Let $e=3$ satisfying $\mathrm{Tr}_n^{2n}(\beta^3)=0$ Then $\mathcal{C}$ is a $(4,4,5,4)$-DRCSS. According to Theorem \ref{theorem4}, we obtain a DRCSS $\mathcal{C}$ are known in Table \ref{table8}, where each entry represents a power of $\xi_4$. By MATLAB problem, we respectively show the auto-ambiguity magnitude distribution of $\mathbf{C}^1$, the auto-ambiguity magnitude distribution of $\mathbf{C}^3$ and the cross-ambiguity magnitude distribution of $\mathbf{C}^1$ and $\mathbf{C}^3$ in Fig. \ref{Fig4}.
		\end{example}
		
			\subsection{The fifth construction of aperiodic DRCSSs}
		Let $0 \leq e \leq q$ be an integer satisfying $\mathrm{Tr}_n^{2n}(\beta^e)=0$, where $\beta$ is a primitive element of $\mathbb{F}_{q^2}$.
		Let $\phi(\cdot)$ be an arbitrary one-to-one mapping from $\mathbb{F}_q$ to $\mathbb{Z}_q$. 
		We define a DRCS set
		\begin{equation}\label{equation3.5}
			\mathcal{C}=\{\mathbf{C}^{k}: 0\leq k \leq q-2\},
		\end{equation}
		where $\mathbf{C}^{k}=[\mathbf{c}_0^{k}, \mathbf{c}_1^{k}, \cdots ,\mathbf{c}_{q-1}^{k}]^{\mathrm{T}}$ consist of $K=q-1$ sequences of length $N=q-1$ defined by
		$$\mathbf{c}_m^{k}=\{c_m^{k}(t)\}_{t=0}^{q-2},$$ 
		$$c_m^{k}(t)=\xi_{q-1}^{m\phi(\mathrm{Tr}_n^{2n}(\beta^{e+k(q+1)+t+1}))}, \;\; 0\leq m\leq q-2.$$
		It is straightforward to show that the alphabet size of $\mathcal{C}$ is $q-1$.

		\begin{theorem}\label{theorem5}
			Let $q=p^n>4$, where $p$ is a prime and $n$ is a positive integer. Then $\mathcal{C}$ is an aperiodic $(q-1,q-1,q-1,q-1)$-DRCSS with alphabet size $q-1$ which is asymptotically optimal with respect to the lower bound in \eqref{eqD1}.
		\end{theorem}
		
		\begin{proof}
			The proof is very similar to that of Theorem \ref{theorem4}. We omit it here.
		\end{proof}
		
		\begin{table}[htbp]  
			\centering 
			\caption{The parameters of some aperiodic DRCSSs in Theorem \ref{theorem5} for $q=p$ and $5\leq q \leq 43$} 
			\label{table9} 
			\begin{tabular}{ccccccc} 
				\toprule 
				q     & K & M & N &$\hat{\theta}_\text{max}$  & $\hat{\theta}_{\text{opt}}$ & $\hat{\rho}$   \\
				\midrule 
				5     & 4 & 4 & 4 &4         & 2.6005 & 1.5382 \\
				7     & 6 & 6 & 6 &6         & 4.3623 & 1.3754 \\
				11    & 10 & 10 & 10 &10     & 7.9678 & 1.2551 \\
				13    & 12 & 12 & 12 &12     & 9.7980 & 1.2247\\
				17    & 16 & 16 & 16 &16     & 13.4944 & 1.1857 \\
				19    & 18 & 18 & 18 &18     & 15.3563 & 1.1722\\
				23    & 22 & 22 & 22 &22     & 19.1010 & 1.1518 \\
				29    & 28 & 28 & 28 &28     & 24.7572 & 1.1310 \\
				31    & 30 & 30 & 30 &30     & 26.6508 & 1.1257 \\
				37    & 36 & 36 & 36 &36     & 32.3510 & 1.1128\\
				41    & 40 & 40 & 40 &40     & 36.1646 & 1.1061 \\
				43    & 42 & 42 & 42 &42     & 38.0749 & 1.1031 \\
				\bottomrule 
			\end{tabular}
		\end{table}
		
		\begin{remark}
			In Table \ref{table9}, we list the parameters of some aperiodic DRCSSs $\mathcal{C}$ constructed in Theorem \ref{theorem5}. The numerical date show that the optimal factor of $\mathcal{C}$ approaches $1$ quickly when $q$ increases.
		\end{remark}

		\begin{table}[htbp]
			\centering
			\caption{The DRCS Set $\mathcal{C}$ in Example \ref{example5}}
			\label{table10}
			\begin{tabular}{|c|c|c|c|}
				\hline
				$\mathbf{C}^{0}$ & $\mathbf{C}^{1}$  & $\mathbf{C}^{2}$ & $\mathbf{C}^{3}$ \\
				\hline
				$\begin{bmatrix}          
					0 0 0 0    \\
					1 0 0 3   \\
					2 0 0 2   \\				
					3 0 0 1  
				\end{bmatrix}$
				&
				$\begin{bmatrix}
					0 0 0 0   \\
					2 3 3 1   \\
					0 2 2 2  \\
					2 1 1 3  
				\end{bmatrix}$
				&
				$\begin{bmatrix}
					0 0 0 0   \\
					0 1 1 2  \\
					0 2 2 0 \\
					0 3 3 2  
				\end{bmatrix}$ 
				&
				$\begin{bmatrix}
					0 0 0 0    \\
					3 2 2 0   \\
					2 0 0 0  \\
					1 2 2 0   
				\end{bmatrix}$\\
				\hline
			\end{tabular}
		\end{table}
		
		\begin{figure}[!htbp]  
			\centering
			\includegraphics[width=0.9\linewidth]{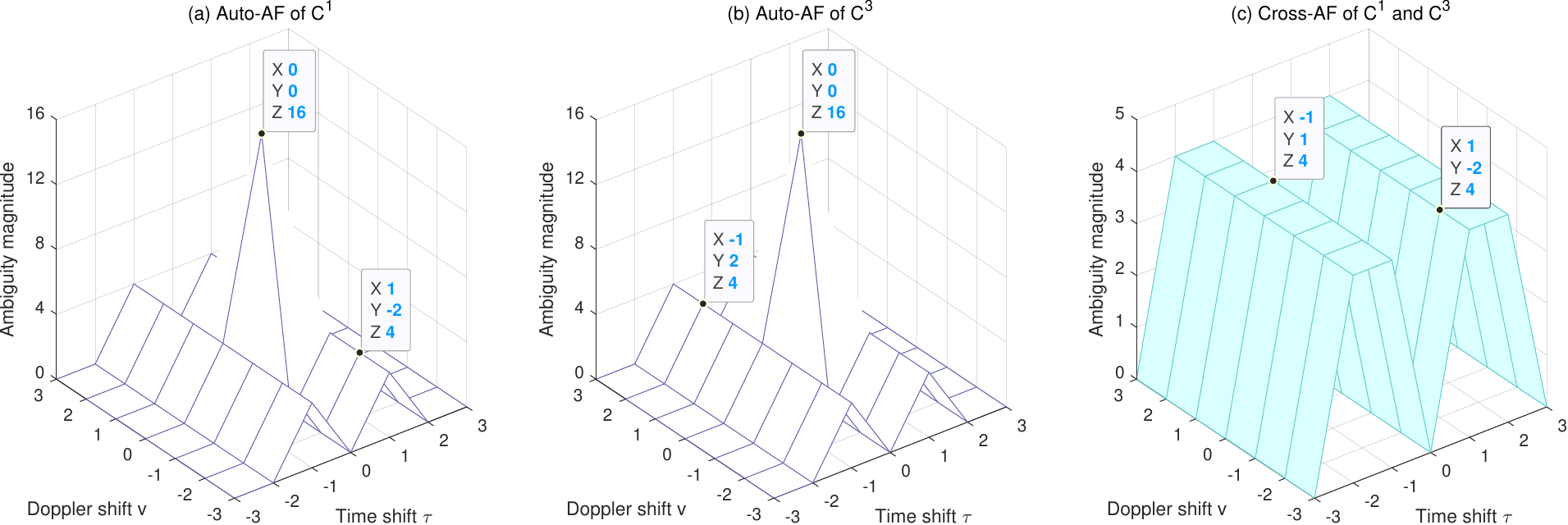}
			\caption{The aperiodic auto-ambiguity and cross-ambiguity magnitude of $\mathbf{C}^{1}$ and $\mathbf{C}^{3}$ in Example \ref{example5}}
			\label{Fig5}  
		\end{figure}
		
		\begin{example}\label{example5}
			Let $q=5$ and $\beta$ be a primitive element of $\mathbb{F}_{5^2}$ satisfying $\beta^2+\beta+2=0$. Consider a one-to-one mapping function from $\mathbb{F}_5$ to $\mathbb{Z}_5$, defined by $\phi(x)=x$ for $x \in \mathbb{F}_5$. Let $e=3$ satisfying $\mathrm{Tr}_n^{2n}(\beta^3)=0$ Then $\mathcal{C}$ is a $(4,4,4,4)$-DRCSS. According to Theorem \ref{theorem5}, we obtain a DRCSS $\mathcal{C}$ are known in Table \ref{table10}, where each entry represents a power of $\xi_4$. By MATLAB problem, we respectively show the auto-ambiguity magnitude distribution of $\mathbf{C}^1$, the auto-ambiguity magnitude distribution of $\mathbf{C}^3$ and the cross-ambiguity magnitude distribution of $\mathbf{C}^1$ and $\mathbf{C}^3$ in Fig. \ref{Fig5}.
		\end{example}
	
	\section{Conclusion}\label{s4}
	In this paper, based on trace functions over finite fields and column orthogonal complex matrices, we present five classes of aperiodic DRCSSs with
	novel parameters, all of which are asymptotically optimal with respect to the bound
	given in \eqref{eqD1}. Compared with the existing asymptotically optimal aperiodic DRCSSs reported in the literature, the proposed aperiodic DRCSSs exhibited superior or novel parameters. Notably, for three constructed families of aperiodic DRCSSs, the PAPR of the column sequences can be upper bounded by $p$ by selecting appropriate column orthogonal complex matrices. An interesting direction for future work is to design DRCSSs with novel parameters and new construction methods.

	\section*{Acknowledgments}
	This work was supported by National Natural Science Foundation of China (Gran No. 62272420), Natural Science Foundation of Fujian Province (No. 2023J01535).
	
	\section*{Declarations}
	\textbf{Conflict of interest} The authors declare that they have no conflicts of interest to report regarding the present study.
	
	\bibliographystyle{sn-mathphys-num}

\begin{thebibliography}{99}
	\bibitem{C1}
	Golomb, S. W., Gong, G.: Signal Design for Good Correlation: for Wireless Communication, Cryptography, and Radar. Cambridge University Press, Cambridge (2005)
	
	\bibitem{C2}
	Fan, P. and Darnell, M.: Sequence design for communications applications. Research Studies Press, (1996)
	
	\bibitem{C3}
	Welch, L.: Lower bounds on the maximum cross correlation of signals (Corresp.). IEEE Transactions on Information Theory. {\bf 20}(3), 397--399 (1974)
	
	\bibitem{C4}
	Golay, M.: Complementary series. IRE Transactions on Information Theory. {\bf 7}(2), 82--87 (1961)
	
	\bibitem{C5}
	Tseng, C.C., Liu, C.: Complementary sets of sequences. IEEE Transactions on Information Theory. {\bf 18}(5), 644--652 (1972)
	
	\bibitem{C6}
	Bomer, L., Antweiler, M.:  Periodic complementary binary sequences. IEEE Transactions on Information Theory. {\bf 36}(6), 1487--1494 (1990)
	
	\bibitem{C7}
	Luke, H.D.: Binary odd-periodic complementary sequences. IEEE Transactions on Information Theory. {\bf 43}(1), 365--367 (1997)
	
	\bibitem{C8}
	Suehiro, N., Hatori, M.: N-shift cross-orthogonal sequences. IEEE Transactions on Information Theory. {\bf 34}(1), 143--146 (1988)
	
	\bibitem{C9}
	Chen, H.H., Yeh, J.F., Suehiro, N.: A multicarrier CDMA architecture based on orthogonal complementary codes for new generations of wideband wireless communications. IEEE Communications Magazine. {\bf 39}(10), 126--135 (2001)
	
	\bibitem{C10}
	Sojevic, P., Georghiades, C.N.: Complementary sequences for ISI channel estimation. IEEE Transactions on Information Theory. {\bf 47}(3), 1145--1152 (2001)
	
	\bibitem{C11}
	Pezeshki, A., Calderbank, A. R., Moran, W., Howard, S. D.: Doppler Resilient Golay Complementary Waveforms. IEEE Transactions on Information Theory. {\bf 54}(9), 4254--4266 (2008)  
	
	\bibitem{C12}
	Davis, J.A., Jedwab, J.: Peak-to-mean power control in OFDM, Golay complementary sequences, and Reed-Muller codes. IEEE Transactions on Information Theory. {\bf 45}(7), 2397--2417 (1999)
	
	\bibitem{C13}
	Fan, P., Yuan, W., Tu, Y.: Z-complementary Binary Sequences. IEEE Signal Processing Letters. {\bf 14}(8), 509--512 (2007)
	
	\bibitem{C14}
	Li, J., Huang, A., Mohsen, G., Chen, H.: Inter-Group Complementary Codes for Interference-Resistant CDMA Wireless Communications. IEEE Transactions on Wireless Communications. {\bf 7}(1), 166--174 (2008)
	
	\bibitem{C15}
	Liu, Z., Guan, Y., Ng, B., Chen, H.H.: Correlation and Set Size Bounds of Complementary Sequences with Low Correlation Zone. IEEE Transactions on Communications. {\bf 59}(12), 3285--3289 (2011)
	
	\bibitem{C16}
	Liu, Z., Parampalli, U., Guan, Y., Boztas, S.: Constructions of Optimal and Near-Optimal Quasi-Complementary Sequence Sets from Singer Difference Sets. IEEE Wireless Communications Letters. {\bf 2}(5), 487--490 (2013)
	
	\bibitem{C17}
	Liu, Z., Guan, Y., Mow, W. H.: A Tighter Correlation Lower Bound for Quasi-Complementary Sequence Sets. IEEE Transactions on Information Theory. {\bf 60}(1), 388--396 (2014)		
	
	
	
	\bibitem{C18}
	Ding, C., Feng, K., Feng, R.: Unit time-phase signal sets: Bounds and constructions. Cryptography and Communications. {\bf 5}(3), 209--227 (2013)
	
	\bibitem{C19}
	Ye, Z., Zhou, Z., Fan, P., Liu, Z., Lei, X., Tang, X.: Low Ambiguity Zone: Theoretical Bounds and Doppler-Resilient Sequence Design in Integrated Sensing and Communication Systems. IEEE Journal on Selected Areas in Communications. {\bf 40}(6), 1809--1822 (2022)
	
	\bibitem{C20}
	Tian, L., Song, X., Liu, Z., Li, Y.: Asymptotically Optimal Sequence Sets With Low/Zero Ambiguity Zone Properties. IEEE Transactions on Information Theory. {\bf 71}(6), 4785--4796 (2025)
	
	\bibitem{C21}
	Wang, Z., Shen, B., Yang, Y., Zhou, Z.: New Construction of Asymptotically Optimal Low Ambiguity Zone Sequence Sets. IEEE Signal Processing Letters. {\bf 32}, 2609--2613 (2025)
	
	\bibitem{C22}
	Wang, Z., Zhou, Z., Adhikary, A. R., Yang, Y., Mesnager, S., Fan, P.: Asymptotically Optimal Aperiodic and Periodic Sequence Sets with Low Ambiguity Zone Through Locally Perfect Nonlinear Functions. IEEE Transactions on Information Theory. (2026) DOI: 10.1109/TIT.2026.3671733
	
	\bibitem{C23}
	Yang, Z., Wang, Z., Liu, H., Feng, K.: New Constructions of Locally Perfect Nonlinear Functions and Their Application to Sequence Sets With Low Ambiguity Zone. IEEE Transactions on Information Theory. (2026) DOI: 10.1109/TIT.2026.3669567
	
	\bibitem{C24}
	Peng, X., Cheng, J., Wu, C., Li, C., Liu, Z.: New Constructions of Asymptotically Optimal Zero/Low Ambiguity Zone Sequence Sets. IEEE Signal Processing Letters. {\bf 33}, 808--812 (2026)
	
	\bibitem{C25}
	Shen, B., Yang, Y., Zhou, Z., Liu, Z., Fan, P.: Doppler Resilient Complementary Sequences: Theoretical Bounds and Optimal Constructions. IEEE Transactions on Information Theory. {\bf 71}(7), 5166--5177 (2025)
	
	\bibitem{C26}
	Shen, B., Zhou, Z., Yang, Y., Fan, P.: Some Optimal and Near Optimal Doppler Resilient Complementary Sequence Sets. arXiv preprint {\bf arXiv:2508.15325} (2025)
		
	\bibitem{D1}
	Wang, Z., Yang, Y., Zhou, Z., Adhikary, A. R., Fan, P.: Doppler Resilient Complementary Sequences: Tighter Aperiodic Ambiguity Function Bound and Optimal Constructions. arXiv preprint. {\bf arXiv:2505.11012} (2025)
		
		\bibitem{D2}
		Wang, Z., Yang, Z., Yang, Y., Adhikary, A. R., Feng, K.: Asymptotically Optimal Aperiodic Doppler Resilient Complementary Sequence Sets Via Generalized Quasi-Florentine Rectangles. arXiv preprint. {\bf arXiv:2602.06045} (2025)
		
	
	\bibitem{a1}
	Li, Y., Tian, L., Xu, C.: Constructions of Asymptotically Optimal Aperiodic Quasi-Complementary Sequence Sets. IEEE Transactions on Communications. {\bf 67}(11), 7499--7511 (2019)
	
	\bibitem{a2}
	Xiao, H., Luo, G., Cao, X.: New Constructions of Asymptotically Optimal Quasi-Complementary Sequence Sets With Small Alphabet Sizes. IEEE Transactions on Communications. {\bf 73}(8), 5881--5890 (2025)
	
	\bibitem{a3}
	Wang, P., Heng, Z., Li, C.: New Constructions of Asymptotically Optimal Periodic and Aperiodic Quasi-Complementary Sequence Sets. IEEE Transactions on Communications. {\bf 73}(12), 14167--14182 (2025)
	
		\bibitem{R20}
		Liu, Z., Guan, Y., Parampalli, U.: New Complete Complementary Codes for Peak-to-Mean Power Control in Multi-Carrier CDMA. IEEE Transactions on Communications. {\bf 62}(3), 1105--1113 (2014)
	
	\bibitem{R21}
	Simon, M. K., Omura, J. K., Scholtz, R. A., Levitt, B. K.: Spread Spectrum Communications Handbook. Rockville, MD, USA: Computer Science (1985)
	
	
		\bibitem{q1}
		Yu, N. Y.: Non-Orthogonal Golay-Based Spreading Sequences for Uplink Grant-Free Access. IEEE Communications Letters. {\bf 24}(10), 2104--2108 (2020)
		
		\bibitem{q2}
		Yu, N. Y.: Binary Golay Spreading Sequences and Reed-Muller Codes for Uplink Grant-Free NOMA. IEEE Transactions on Communications. {\bf 69}(1), 276--290 (2021)
		
		\bibitem{q3}
		Tian, L., Liu, T., Li, Y.: New Constructions of Binary Golay Spreading Sequences for Uplink Grant-Free NOMA. IEEE Communications Letters. {\bf 26}(10), 2480--2484 (2022)
		
		\bibitem{q4}
		Liu, K., Zhou, Z., Adhikary, A. R., Tang, C.: Large Sets of Binary Spreading Sequences With Low Correlation and Low PAPR via Gold Functions. IEEE Transactions on Information Theory. {\bf 70}(7), 5309--5322 (2024)
		
		\bibitem{q5}
		Liu, K., Zhou, Z., Adhikary, A. R., Luo, R: New sets of non-orthogonal spreading sequences with low correlation and low PAPR using extended Boolean functions. Designs, Codes and Cryptography. {\bf 91}(10), 3115--3139 (2023)
		
		\bibitem{q6}
		Xiang, C., Tang, C., Qiu, W.: Three new classes of spreading sequence sets with low correlation and PAPR. Finite Fields and Their Applications. {\bf 103}, 102575 (2025)

		
		
		
		
		
	\end{thebibliography}

\end{document}